\documentclass[envcountsame,oribibl]{llncs}
\usepackage[utf8]{inputenc}
\usepackage[vlined,linesnumbered]{algorithm2e}
\usepackage{booktabs}
\usepackage{multirow}
\usepackage{amssymb}
\usepackage{tikz}
\usepackage{wrapfig}
\usepackage{caption}
\usepackage{subcaption}
\usepackage{rotating}
\newcommand{\PseudoCount}{p_{\mathrm{c}}}
\newcommand{\ChildCloseness}{\mathrm{child}_{\mathrm{close}}}
\newcommand{\BestSol}{\mathrm{score}_{\mathrm{max}}}
\newcommand{\DFS}{\mathrm{DFS}}

\usepackage[defernumbers=false,firstinits=true,url=false,maxbibnames=99]{biblatex}
\DeclareFieldFormat{labelnumberwidth}{\mkbibbold{#1.}}
\DeclareNameAlias{sortname}{last-first}
\DeclareNameAlias{default}{last-first}
\begin{document}
\mainmatter

\begin{refsegment}

\title{Fast Quasi-Threshold Editing\thanks{This work was supported by the DFG under grants BR~2158/6-1, WA~654/22-1, and BR~2158/11-1}}

\author{Ulrik~Brandes\inst{1} \and Michael~Hamann\inst{2} \and Ben~Strasser\inst{2} \and Dorothea~Wagner\inst{2}}

\institute{Computer \& Information Science, University of Konstanz, Germany\\
\email{ulrik.brandes@uni-konstanz.de}
\and Faculty of Informatics, Karlsruhe Institute of Technology, Germany\\
\email{\{michael.hamann,strasser,dorothea.wagner\}@kit.edu}
}

\maketitle

\begin{abstract}
We introduce Quasi-Threshold Mover (QTM), an algorithm to solve the quasi-threshold (also called trivially perfect) graph editing problem with edge insertion and deletion.
Given a graph it computes a quasi-threshold graph which is close in terms of edit count. 
This edit problem is NP-hard.
We present an extensive experimental study, in which we show that QTM is the first algorithm that is able to scale to large real-world graphs in practice. 
As a side result we further present a simple linear-time algorithm for the quasi-threshold recognition problem.
\end{abstract}

\section{Introduction}

\begin{wrapfigure}[14]{o}{0.25\textwidth}
  \vspace{-3em}
  \begin{center}
    \includegraphics[scale=1]{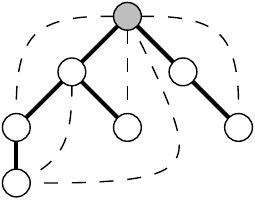}
  \end{center}
  \caption{\label{qt-example}Quasi-thres. graph with thick skeleton, grey root and dashed transitive closure.}
\end{wrapfigure}

Quasi-Threshold graphs, also known as \emph{trivially perfect} graphs, are defined as the $P_4$- and $C_4$-free graphs, i.e., the graphs that do not contain a path or cycle of length~4 as node-induced subgraph~\cite{ycc-q-96}.
They can also be characterized as the transitive closure of rooted forests \cite{w-antcg-65}, as illustrated in Figure~\ref{qt-example}.
These forests can be seen as skeletons of quasi-threshold graphs.
Further a constructive characterization exists: Quasi-threshold graphs are the graphs that are closed under disjoint union and the addition of isolated nodes and nodes connected to every existing node~\cite{ycc-q-96}.

Linear time quasi-threshold recognition algorithms were proposed in~\cite{ycc-q-96} and in~\cite{c-albfs-08}.
Both construct a skeleton if the graph is a quasi-threshold graph.
Further, \cite{c-albfs-08} also finds a $C_4$ or $P_4$ if the graph is no quasi-threshold graph.

Nastos and Gao \cite{ng-f-13} observed that components of quasi-threshold graphs have many features in common with the informally defined notion of communities in social networks.
They propose to find a quasi-threshold graph that is close to a given graph in terms of edge edit distance in order to detect the communities of that graph.
Motivated by their insights we study the quasi-threshold graph editing problem in this paper.
Given a graph $G=(V,E)$ we want to find a quasi-threshold graph $G'=(V,E')$ which is closest to $G$, i.e., we want to minimize the number $k$ of edges in the symmetric difference of $E$ and $E'$.
Figure~\ref{qt-edit-example} illustrates %
\begin{wrapfigure}[15]{o}{0.35\textwidth}
\begin{center}
\includegraphics[scale=1]{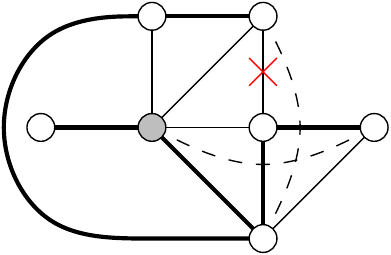}
\end{center}
\caption{\label{qt-edit-example}Edit example with solid input edges, dashed inserted edges, a crossed deleted edge, a thick skeleton with grey root.}%
\end{wrapfigure}
an edit example.
Unfortunately, the quasi-threshold graph editing problem is NP-hard~\cite{ng-f-13}.
However, the problem is fixed parameter tractable (FPT) in $k$ as it is defined using forbidden subgraphs~\cite{c-f-96}.
A basic bounded search tree algorithm which tries every of the 6 possible edits of a forbidden subgraph has a running time in $O(6^k \cdot (|V|+|E|))$.
In \cite{dp-apktp-14} a polynomial kernel of size $O(k^7)$ was introduced.
Unfortunately, our experiments show that real-world social networks have a prohibitively large amount of edits.
We prove lower bounds on real-world graphs for $k$ on the scale of $10^4$ and $10^5$.
A purely FPT-based algorithm with parameter $k$ can thus not scale in practice.
The only heuristic we are aware of was introduced by Nastos and Gao~\cite{ng-f-13} but it examines all $\Theta(|V|^2)$ possible edits in each greedy editing step and thus needs time $\Omega(k\cdot|V|^2)$.
Even though this running time is polynomial it is still prohibitive for large graphs.
In this paper we fill this gap by introducing Quasi-Threshold Mover (QTM), the first scalable quasi-threshold editing algorithm.
The final aim of our research is to determine whether quasi-threshold editing is a useful community detection algorithm.
Designing an algorithm able of solving the quasi-threshold editing problem on large real-world graphs is a first step in this direction.

\subsection{Our Contribution}

Our main contribution is Quasi-Threshold Mover (QTM), a scalable quasi-threshold editing algorithm.
We provide an extensive experimental evaluation on generated as well as a variety of real-world graphs.
We further propose a simplified certifying quasi-threshold recognition algorithm.
QTM works in two phases: An initial skeleton forest is constructed by a variant of our recognition algorithm, and then refined by moving one node at a time to reduce the number of edits required.
The running time of the first phase is dominated by the time needed to count the number of triangles per edge. 
The best current triangle counting algorithms run in $O(|E|\alpha(G))$ \cite{cn-asla-85,ob-tlabd-14} time, where $\alpha(G)$ is the arboricity.
These algorithms are efficient and scalable in practice on the considered graphs.
One round of the second phase needs $O(|V|+|E|\log \Delta)$ time, where $\Delta$ is the maximum degree.
We show that four rounds are enough to achieve good results.

\subsection{Preliminaries}

We consider simple, undirected graphs $G = (V, E)$ with $n = |V|$ nodes and $m = |E|$ edges.
For $v \in V$ let $N(v)$ be the adjacent nodes of $v$.
Let $d(v) := |N(v)|$ for $v \in V$ be the degree of $v$ and $\Delta$ the maximum degree in $G$.
Whenever we consider a skeleton forest, we denote by $p(u)$ the parent of a node $u$.

\section{Lower Bounds}

A lot of previous research has focused on FPT-based algorithms.
To show that no purely FPT-based algorithm parameterized in the number of edits can solve the problem we compute lower bounds on the number of edits required for real-world graphs. 
The lower bounds used by us are far from tight. 
However, the bounds are large enough to show that any algorithm with a running time superpolynomial in $k$ can not scale.

To edit a graph we must destroy all forbidden subgraphs $H$. 
For quasi-threshold editing $H$ is either a $P_4$ or a $C_4$.
This leads to the following basic algorithm: Find forbidden subgraph $H$, increase the lower bound, remove all nodes of $H$, repeat.
This is correct as at least one edit incident to $H$ is necessary. 
If multiple edits are needed then accounting only for one is a lower bound.
We can optimize this algorithm by observing that not all nodes of $H$ have to be removed.
If $H$ is a $P_4$ with the structure $A-B-C-D$ it is enough to remove the two central nodes $B$ and $C$.
If $H$ is a $C_4$ with nodes $A$, $B$, $C$, and $D$ then it is enough to remove two adjacent nodes.
Denote by $B$ and $C$ the removed nodes.
This optimization is correct if at least one edit incident to $B$ or $C$ is needed.
Regardless of whether $H$ is a $P_4$ or a $C_4$ the only edit not incident to $B$ or $C$ is inserting or deleting $\{A, D\}$.
However, this edit only transforms a $P_4$ into a $C_4$ or vice versa.
A subsequent edit incident to $B$ or $C$ is thus necessary.

$H$ can be found using the recognition algorithm.
However, the resulting running time of $O(k(n+m))$ does not scale to the large graphs.
In the appendix we describe a running time optimization to accelerate computations.

\section{Linear Recognition and Initial Editing}
\label{sec:linear_recognition}

The first linear time recognition algorithm for quasi-threshold graphs was proposed in \cite{ycc-q-96}.
In \cite{c-albfs-08}, a linear time certifying recognition algorithm based on lexicographic breadth first search was presented.
However, as the authors note, sorted node partitions and linked lists are needed, which result in large constants behind the big-O.
We simplify their algorithm to only require arrays but still provide negative and positive certificates.
Further we only need to sort the nodes once to iterate over them by decreasing degree.
Our algorithm constructs the forest skeleton of a graph $G$. 
If it succeeds $G$ is a quasi threshold graph and outputs for each node $v$ a parent node $p(v)$.
If it fails it outputs a forbidden subgraph $H$.

To simplify our algorithm we start by adding a super node $r$ to $G$ that is connected to every node and obtain $G'$.
$G$ is a quasi threshold graph if and only if $G'$ is one.
As $G'$ is connected its skeleton is a tree.
A core observation is that higher nodes in the tree must have higher degrees, i.e., $d(v)\le d(p(v))$.
We therefore know that $r$ must be the root of the tree. 
Initially we set $p(u)=r$ for every node $u$.
We process all remaining nodes ordered decreasingly by degree.
Once a node is processed its position in the tree is fixed.
Denote by $u$ the node that should be processed next.
We iterate over all non-processed neighbors $v$ of $u$ and check whether $p(u)=p(v)$ holds and afterwards set $p(v)$ to $u$.
If $p(u)=p(v)$ never fails then $G$ is a quasi-threshold graph as for every node $x$ (except $r$) we have that by construction that the neighborhood of $x$ is a subset of the one of $p(x)$.
If $p(u)\neq p(v)$ holds at some point then a forbidden subgraph $H$ exists.
Either $p(u)$ or $p(v)$ was processed first.
Assume without lose of generality that it was $p(v)$.
We know that no edge $(v, p(u))$ can exist because otherwise $p(u)$ would have assigned itself as parent of $v$ when it was processed.
Further we know that $p(u)$'s degree can not be smaller than $u$'s degree as $p(u)$ was processed before $u$.
As $v$ is a neighbor of $u$ we know that another node $x$ must exist that is a neighbor of $p(u)$ but not of $u$, i.e., $(u, x)$ does not exist.
The subgraph $H$ induced by the 4-chain $v-u-p(u)-x$ is thus a $P_4$ or $C_4$ depending on whether the edge $(v, x)$ exists.
We have that $u\neq r$ as $u$ is processed by the algorithm and $v\neq r$ as its degree is at most $d(u)$.
Further $p(u)\neq r$ as $p(v)$ was processed before $p(u)$ and $x\neq r$ as $r$ is a neighbor of $u$.
$H$ therefore does not use $r$ and is contained in $G$.

\paragraph{From Recognition to Editing.}
\label{sec:linear_editing}

We modify the recognition algorithm to construct a skeleton for arbitrary graphs. 
This skeleton induces a quasi threshold graph $Q$.
We want to minimize $Q$'s distance to $G$.
Note that all edits are performed implicitly, we do not actually modify the input graph for efficiency reasons.
The only difference between our recognition and our editing algorithm is what happens when we process a node $u$ that has a non-processed neighbor $v$ with $p(u)\neq p(v)$.
The recognition algorithm constructs a forbidden subgraph $H$, while the editing algorithm tries to resolve the problem.
We have three options for resolving the problem: we ignore the edge $\{u, v\}$, we set $p(v)$ to $p(u)$, or we set $p(u)$ to $p(v)$.
The last option differs from the first two as it affects all neighbors of $u$.
The first two options are the decision if we want to make $v$ a child of $u$ even though $p(u) \neq p(v)$ or if we want to ignore this potential child.
We start by determining a preliminary set of children by deciding for each non-processed neighbor of~$u$ whether we want to keep or discard it.
These preliminary children elect a new parent by majority.
We set $p(u)$ to this new parent.
Changing~$u$'s parent can change which neighbors are kept.
We therefore reevaluate all the decisions and obtain a final set of children for which we set $u$ as parent.
Then the algorithm simply continues with the next node.

What remains to describe is when our algorithm keeps a potential child.
It does this using two edge measures: The number of triangles $t(e)$ in which an edge $e$ participates and a pseudo-$C_4$-$P_4$-counter $\PseudoCount(e)$, which is the sum of the number of $C_4$ in which $e$ participates and the number of $P_4$ in which $e$ participates as central edge.
Computing $\PseudoCount(x,y)$ is easy given the number of triangles and the degrees of $x$ and $y$ as  $\PseudoCount(\{x, y\}) = (d(x) - 1 - t(\{x, y\}))\cdot(d(y) - 1 - t(\{x, y\}))$ holds.
Having a high $\PseudoCount(e)$ makes it likely that $e$ should be deleted.
We keep a potential child only if two conditions hold.
The first is based on triangles.
We know by construction that both $u$ and $v$ have many edges in $G$ towards their current ancestors.
Keeping $v$ is thus only useful if $u$ and $v$ share a large number of ancestors as otherwise the number of induced edits is too high.
Each common ancestor of $u$ and $v$ results in a triangle involving the edge $\{u,v\}$ in $Q$.
Many of these triangles should also be contained in $G$.
We therefore count the triangles of $\{u,v\}$ in $G$ and check whether there are at least as many triangles as $v$ has ancestors.
The other condition uses $\PseudoCount(e)$. 
The decision whether we keep $v$ is in essence the question of whether $\{u, v\}$ or $\{v, p(v)\}$ should be in $Q$.
We only keep $v$ if $\PseudoCount(\{u, v\})$ is not higher than $\PseudoCount(\{v, p(v)\})$.
The details of the algorithm can be found in the appendix.
The time complexity of this heuristic editing algorithm is dominated by the triangle counting algorithm as the rest is linear.

\section{The Quasi-Threshold Mover Algorithm}
\label{sec:local_moving}

\begin{figure}[tb]
\centering
\begin{subfigure}[b]{0.68\textwidth}%
\begin{algorithm}[H]
\ForEach{$v_m$-neighbor $u$}{
	push $u$\;
}
\While{queue not empty}{
	$u\leftarrow$ pop\;
	determine $\ChildCloseness(u)$ by DFS\;
	$x\leftarrow \max$ over $\BestSol$ of reported $u$-children\;
	$y\leftarrow \sum$ over $\ChildCloseness$ of close $u$-children\;
	
	\uIf{$u$ is $v_m$-neighbor}{
		$\BestSol(u)\leftarrow \max\{x,y\} + 1$\;
	}
	\Else{
		$\BestSol(u)\leftarrow \max\{x,y\} - 1$\;
	}
	\If{$\ChildCloseness(u)>0$ or $\BestSol(u)>0$}{
		report $u$ to $p(u)$\;
		push $p(u)$\;
	}
}
Best $v_m$-parent corresponds to $\BestSol(r)$\;
\end{algorithm}%
\caption{Pseudo-Code for moving $v_m$}%
\label{fig:moving-v-m-pseudo-code}
\end{subfigure}%
~ 
\begin{subfigure}[b]{0.28\textwidth}
	\centering
	\includegraphics[width=2.2cm]{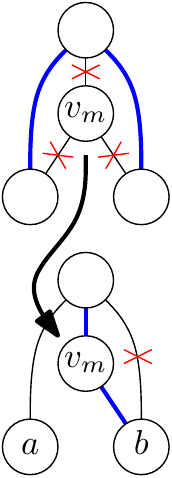}
	\caption{ Moving $v_m$ example. }
	\label{fig:move-v-m}
\end{subfigure}

\caption{In Figure \ref{fig:move-v-m} the drawn edges are in the skeleton. Crossed edges are removed while thick blue edges are inserted by moving $v_m$. $a$ is not adopted while $b$ is.}
\end{figure}

Our algorithm iteratively increases the quality of a skeleton $T$ using an algorithm based on local moving.
Local moving is a successful technique that is employed in many heuristic community detection algorithms \cite{bgll-f-08,gkw-edcgc-14,rn-m-11}.
As in most algorithm based on this principle, our algorithm works in rounds.
In each round it iterates over all nodes $v_m$ in random order and tries to move $v_m$.
In the context of community detection, a node is moved to a neighboring community such that a certain objective function is increased.
In our setting we want to minimize the number of edits needed to transform the input graph $G$ into the quasi-threshold graph $Q$ implicitly defined by $T$.
We need to define the set of allowed moves for $v_m$ in our setting.
Moving $v_m$ consists of moving $v_m$ to a different position within $T$ and is illustrated in Figure~\ref{fig:move-v-m}.
We need to chose a new parent $u$ for $v_m$.
The new parent of $v_m$'s old children is $v_m$'s old parent.
Besides choosing the new parent $u$ we select a set of children of $u$ that are \emph{adopted} by $v_m$, i.e., their new parent becomes $v_m$.
Among all allowed moves for $v_m$ we chose the move that reduces the number of edits as much as possible.
Doing this in sub-quadratic running time is difficult as $v_m$ might be moved anywhere in $G$.
By only considering the neighbors of $v_m$ in $G$ and a few more nodes per neighbor in a bottom-up scan in the skeleton, our algorithm has a running time in $O(n+m \log \Delta)$ per round.
While our algorithm is not guaranteed to be optimal as a whole we can prove that for each node $v_m$ we choose a move that reduces the number of edits as much as possible.
Our experiments show that given the result of the initialization heuristic our moving algorithm performs well in practice.
They further show that in practice four rounds are good enough which results in a near-linear total running time.

\paragraph{Basic Idea.}

Our algorithm starts by isolating $v_m$, i.e., removing all incident edges in $Q$.
It then finds a position at which $v_m$ should be inserted in $T$.
If $v_m$'s original position was optimal then it will find this position again.
For simplicity we will assume again that we add a virtual root $r$ that is connected to all nodes.
Isolating $v_m$ thus means that we move $v_m$ below the root $r$ and do not adopt any children.
Choosing $u$ as parent of $v_m$ requires $Q$ to contain edges from all ancestors of $u$ to $v_m$.
Further if $v_m$ adopts a child $w$ of $u$ then $Q$ must have an edge from every descendant of $w$ to $v_m$.
How good a move is depends on how many of these edges already exist in $G$ and how many edges incident to $v_m$ in $G$ are not covered.
To simplify notation we will refer to the nodes incident to $v_m$ in $G$ as \emph{$v_m$-neighbors}.
We start by identifying which children a node should adopt.
For this we define the \emph{child closeness} $\ChildCloseness(u)$ of $u$ as the number of $v_m$-neighbors in the subtree of $u$ minus the non-$v_m$-neighbors.
A node $u$ is a \emph{close child} if $\ChildCloseness(u)>0$.
If $v_m$ chooses a node $u$ as new parent then it should adopt all close children.
A node can only be a close child if it is a neighbor of $v_m$ or when it has a close child.
Our algorithm starts by computing all close children and their closeness using many short DFS searches in a bottom up fashion.
Knowing which nodes are good children we can identify which nodes are good parents for $v_m$.
A potential parent must have a close child or must be a neighbor of $v_m$.
Using the set of close children we can easily derive a set of parent candidates and an optimal selection of adopted children for every potential parent.
We need to determine the candidate with the fewest edits.
We do this in a bottom-up fashion.%
To implement the described moving algorithm we need to put $O(d_G(v_m))$ elements into a priority queue. 
The running time is thus amortized $O(d_G(v_m) \log d_G(v_m))$ per move or $O(n+m \log \Delta)$ per round.
We start many small searches and analyze their running time complexity using tokens. 
Initially only the $v_m$-neighbors have tokens.
A search consumes a token per step.
The details of the analysis are complex and are described in the appendix.

\paragraph{Close Children.}

To find all close children we attach to each node $u$ a DFS instance that explores the subtree of $u$.
Note that every DFS instance has a constant state size and thus the memory consumption is still linear.
$u$ is close if this DFS finds more $v_m$-neighbors than non-$v_m$-neighbors.
Unfortunately we can not fully run all these searches as this requires too much running time.
Therefore a DFS is aborted if it finds more non-$v_m$-neighbors than $v_m$-neighbors.
We exploit that close children are $v_m$-neighbors or have themselves close children.
Initially we fill a queue of potential close children with the neighbors of $v_m$ and when a new close child is found we add its parent to the queue.
Let $u$ denote the current node removed from the queue.
We run $u$'s DFS and if it explores the whole subtree then $u$ is a close child.
We need to take special care that every node is visited only by one DFS.
A DFS therefore looks at the states of the DFS of the nodes it visits. 
If one of these other DFS has run then it uses their state information to skip the already explored part of the subtree.
To avoid that a DFS is run after its state was inspected we organize the queue as priority queue ordered by tree depth.
If the DFS of $u$ starts by first inspecting the wrong children then it can get stuck because it would see the $v_m$-neighbors too late.
The DFS must first visit the close children of $u$.
To assure that $u$ knows which children are close every close child must report itself to its parent when it is detected.
As all children have a greater depth they are detected before the DFS of their parent starts.

\paragraph{Potential Parents.}

Suppose we consider the subtree $T_u$ of $u$ and $w$ is a potential parent in $T_u$. 
Consider the set of nodes $X_w$ given by the ancestors of $w$ and the descendants of all close children of $w$.
$X_w$ includes $w$ and its close children themselves.
Moving $v_m$ below $w$ requires us to insert an edge from $v_m$ to every non-$v_m$-neighbor in $X_w$.
We therefore want $X_w$ to maximize the number of $v_m$-neighbors minus the number of non-$v_m$-neighbors.
This value gives us a score for each potential parent in $T_u$.
We denote by $\BestSol(u)$ the maximum score over all potential parents in $T_u$.
Note that $\BestSol(u)$ is always at least -1 as we can move $v_m$ below $u$ and not adopt any children.
We determine in a bottom-up fashion all $\BestSol(u)$ that are greater than 0.
Whether $\BestSol(u)$ is -1 or 0 is irrelevant because isolating $v_m$ is never worse.
The final solution will be in $\BestSol(r)$ of the root $r$ as its ``subtree'' encompasses the whole graph.
$\BestSol(u)$ can be computed recursively.
If $u$ is a best parent then the value of $\BestSol(u)$ is the sum over the closenesses of all of $u$'s close children $\pm 1$.
If the subtree $T_w$ of a child $w$ of $u$ contains a best parent then $\BestSol(u)=\BestSol(w)\pm 1$.
The $\pm 1$ depends on whether $w$ is a $v_m$-neighbor.
Unfortunately not only potential parents $u$ have a $\BestSol(u)>0$. 
However, we know that every node $u$ with $\BestSol(u)>0$ is a $v_m$-neighbor or has a child $w$ with $\BestSol(w)>0$.
We can therefore process all $\BestSol$ values in a similar bottom-up way using a tree-depth ordered priority queue as we used to compute $\ChildCloseness$.
As both bottom-up procedures have the same structure we can interweave them as optimization and use only a single queue.
The algorithm is illustrated in Figure~\ref{fig:moving-v-m-pseudo-code} in pseudo-code form.

\section{Experimental Evaluation}
\label{sec:exp}

We evaluated the QTM algorithm on the small instances used by Nastos and Gao~\cite{ng-f-13}, on larger generated graphs and large real-world social networks and web graphs.
We measured both the number of edits needed and the required running time.
For each graph we also report the lower bound $b$ of necessary edits that we obtained using our lower bound algorithm.
We implemented the algorithms in C++ using NetworKit~\cite{ssm-nkait-14}.
All experiments were performed on an Intel Core i7-2600K CPU with 32GB RAM.
We ran all algorithms ten times with ten different random node id permutations.

\paragraph{Comparison with Nastos and Gao's Results.}

Nastos and Gao~\cite{ng-f-13} did not report any running times, we therefore re-implemented their algorithm.
Our implementation of their algorithm has a complexity of $O(m^2 + k \cdot n^2 \cdot m)$, the details can be found in the appendix.
Similar to their implementation we used a simple exact bounded search tree (BST) algorithm for the last 10 edits.
In Table~\ref{tab:resultsng} we report the minimum and average number of edits over ten runs.
Our implementation of their algorithm never needs more edits than they reported\footnote{Except on Karate, where they report 20 due to a typo. They also need 21 edits.}.
Often our implementation needs slightly less edits due to different tie-breaking rules.

For all but one graph QTM is at least as good as the algorithm of Nastos and Gao in terms of edits.
QTM needs only one more edit than Nastos and Gao for the grass\_web graph.
The QTM algorithm is much faster than their algorithm, it needs at most 2.5 milliseconds while the heuristic of Nastos and Gao needs up to 6 seconds without bounded search tree and almost 17 seconds with bounded search tree.
The number of iterations necessary is at most 5. As the last round only checks whether we are finished four iterations would be enough.

\begin{table}[htb]
\centering
\caption{Comparison of QTM and \cite{ng-f-13}. We report $n$ and $m$, the lower bound $b$, the number of edits (as minimum, mean and standard deviation), the mean and maximum of number of QTM iterations, and running times in ms.}
\label{tab:resultsng}
\begin{tabular}{lr@{~~}r@{~~}r@{~~}l@{~}|r@{~~}r@{~~}r@{~}|rr@{~}|r@{~~}r}
      Name &      $n$ &      $m$ & $b$ &  Algorithm & \multicolumn{3}{c|}{Edits}  & \multicolumn{2}{c|}{Iterations} & \multicolumn{2}{c}{Time [ms]} \\
           &       &       &             &                             &   min &  mean &  std &              mean &  max &      mean &    std \\
\midrule
  \multirow{3}{*}{dolphins} &  \multirow{3}{*}{62} & \multirow{3}{*}{159} &   \multirow{3}{*}{24} &                       QTM &  72 &  74.1 &  1.1 &               2.7 &  4.0 &      0.6 &     0.1 \\
  &  & & &                   NG w/ BST &  73&  74.7 &  0.9 &                 - &    - &  15\,594.0  & 2\,019.0 \\
  &  & & &                  NG w/o BST &  73&  74.8 &  0.8 &                 - &    - &    301.3  &  4.0  \\[0.3em]
  \multirow{3}{*}{football} & \multirow{3}{*}{115} & \multirow{3}{*}{613} &  \multirow{3}{*}{52} &                       QTM & 251 & 254.3 &  2.7 &               3.5 &  4.0 &    2.5 &   0.4  \\
  & & & &                   NG w/ BST & 255& 255.0 &  0.0 &                 - &    - & 16\,623.3   & 3\,640.6 \\
  & & & &                  NG w/o BST & 255& 255.0 &  0.0 &                 - &    - &  6\,234.6    & 37.7  \\[0.3em]
 \multirow{3}{*}{grass\_web} &  \multirow{3}{*}{86} & \multirow{3}{*}{113} &        \multirow{3}{*}{10} &                       QTM &  35 &  35.2 &  0.4 &               2.0 &  2.0 &      0.5 &     0.1 \\
 &  & & &                   NG w/ BST &  34&  34.6 &  0.5 &                 - &    - &  13\,020.0  & 3\,909.8 \\
 &  & & &                  NG w/o BST &  38&  38.0 &  0.0 &                 - &    - &     184.6  &  1.2  \\[0.3em]
  \multirow{3}{*}{karate} &  \multirow{3}{*}{34} &  \multirow{3}{*}{78} &         \multirow{3}{*}{8} &                       QTM &  21 &  21.2 &  0.4 &               2.0 &  2.0 &      0.4 &     0.1 \\
 &  & & &                   NG w/ BST &  21&  21.0 &  0.0 &                 - &    - &   9\,676.6  &  607.4 \\
 &  & & &                  NG w/o BST &  21&  21.0 &  0.0 &                 - &    - &    28.1   &    0.3 \\[0.3em]
    \multirow{3}{*}{lesmis} &  \multirow{3}{*}{77} & \multirow{3}{*}{254} &        \multirow{3}{*}{13} &                       QTM &  60 &  60.5 &  0.5 &               3.3 &  5.0 &      1.4 & 0.3 \\
  & & & &                   NG w/ BST &  60&  60.8 &  1.0 &                 - &    - &  16\,919.1  & 3\,487.7 \\
  & & & &                  NG w/o BST &  60&  77.1 & 32.4 &                 - &    - &   625.0   &  226.4 \\
\end{tabular}
\end{table}

\paragraph{Large Graphs.}

\begin{table}[htbp]
\caption{Results for large real-world and generated graphs. Number of nodes $n$ and edges $m$, the lower bound $b$ and the number of edits are reported in thousands. Column ``I'' indicates whether we start with a trivial skeleton or not. $\bullet$ indicates an initial skeleton as described in Section~\ref{sec:linear_editing} and $\circ$ indicates a trivial skeleton. Edits and running time are reported for a maximum number of 0 (respectively 1 for a trivial initial skeleton), 4 and $\infty$ iterations. For the latter, the number of actually needed iterations is reported as ``It''. Edits, iterations and running time are the average over the ten runs.}
\label{tab:real-world}
\centering
\begin{tabular}{l@{~}|@{~}lr@{~~}r@{~~}l@{~}|@{~}r@{~~}r@{~~}r@{~}|r@{~}|@{~}r@{~~}r@{~~}r}
& Name & $n$ {[}K{]} & $b$ {[}K{]} & I & \multicolumn{3}{c|}{Edits {[}K{]}} & It & \multicolumn{3}{c}{Time {[}s{]}}\tabularnewline
& & $m$ {[}K{]} &  &  & 0/1 & 4 & $\infty$ & $\infty$ & 0/1 & 4 & $\infty$\tabularnewline
 \midrule
\multirow{14}{*}{\begin{sideways}\hspace{-1em}
Social Networks
\end{sideways}} &
\multirow{2}{*}{Caltech} & 0.77 & \multirow{2}{*}{0.35} & $\bullet$ & 15.8 & 11.6 & 11.6 & 8.5 & 0.0 & 0.0 & 0.1\tabularnewline
 & & 16.66 &  & $\circ$ & 12.6 & 11.7 & 11.6 & 9.4 & 0.0 & 0.0 & 0.1\tabularnewline[0.3em]
 & \multirow{2}{*}{amazon} & 335 & \multirow{2}{*}{99.4} & $\bullet$ & 495 & 392 & 392 & 7.2 & 0.3 & 5.5 & 9.3\tabularnewline
 & & 926 &  & $\circ$ & 433 & 403 & 403 & 8.9 & 1.3 & 4.9 & 10.7\tabularnewline[0.3em]
 & \multirow{2}{*}{dblp} & 317 & \multirow{2}{*}{53.7} & $\bullet$ & 478 & 415 & 415 & 7.2 & 0.4 & 5.8 & 9.9\tabularnewline
 & & 1\,050 &  & $\circ$ & 444 & 424 & 423 & 9.0 & 1.4 & 5.2 & 11.5\tabularnewline[0.3em]
 & \multirow{2}{*}{Penn} & 41.6 & \multirow{2}{*}{19.9} & $\bullet$ & 1\,499 & 1\,129 & 1\,127 & 14.4 & 0.6 & 4.2 & 13.5\tabularnewline
 & & 1\,362 &  & $\circ$ & 1\,174 & 1\,133 & 1\,129 & 16.2 & 1.0 & 3.7 & 14.4\tabularnewline[0.3em]
 & \multirow{2}{*}{youtube} & 1\,135 & \multirow{2}{*}{139} & $\bullet$ & 2\,169 & 1\,961 & 1\,961 & 9.8 & 1.4 & 31.3 & 73.6\tabularnewline
 & & 2\,988 &  & $\circ$ & 2\,007 & 1\,983 & 1\,983 & 10.0 & 7.1 & 28.9 & 72.7\tabularnewline[0.3em]
 & \multirow{2}{*}{lj} & 3\,998 & \multirow{2}{*}{1\,335} & $\bullet$ & 32\,451 & 25\,607 & 25\,577 & 18.8 & 23.5 & 241.9 & 1\,036.0\tabularnewline
 & & 34\,681 &  & $\circ$ & 26\,794 & 25\,803 & 25\,749 & 19.9 & 58.3 & 225.9 & 1\,101.3\tabularnewline[0.3em]
 & \multirow{2}{*}{orkut} & 3\,072 & \multirow{2}{*}{1\,480} & $\bullet$ & 133\,086 & 103\,426 & 103\,278 & 24.2 & 115.2 & 866.4 & 4\,601.3\tabularnewline
 &  & 117\,185 &  & $\circ$ & 106\,367 & 103\,786 & 103\,507 & 30.2 & 187.9 & 738.4 & 5\,538.5\tabularnewline
 \midrule
 \multirow{8}{*}{\begin{sideways}\hspace{-1em}
Web Graphs
\end{sideways}} &
 \multirow{2}{*}{cnr-2000} & 326 & \multirow{2}{*}{48.7} & $\bullet$ & 1\,028 & 409 & 407 & 11.2 & 0.8 & 12.8 & 33.8\tabularnewline
 & & 2\,739 &  & $\circ$ & 502 & 410 & 409 & 10.7 & 3.2 & 11.8 & 30.8\tabularnewline[0.3em]
 & \multirow{2}{*}{in-2004} & 1\,383 & \multirow{2}{*}{195} & $\bullet$ & 2\,700 & 1\,402 & 1\,401 & 11.0 & 7.9 & 72.4 & 182.3\tabularnewline
 & & 13\,591 &  & $\circ$ & 1\,909 & 1\,392 & 1\,389 & 13.5 & 16.6 & 65.0 & 217.6\tabularnewline[0.3em]
 & \multirow{2}{*}{eu-2005} & 863 & \multirow{2}{*}{229} & $\bullet$ & 7\,613 & 3\,917 & 3\,906 & 13.7 & 6.9 & 90.7 & 287.7\tabularnewline
 & & 16\,139 &  & $\circ$ & 4\,690 & 3\,919 & 3\,910 & 14.5 & 22.6 & 85.6 & 303.5\tabularnewline[0.3em]
 & \multirow{2}{*}{uk-2002} & 18\,520 & \multirow{2}{*}{2\,966} & $\bullet$ & 68\,969 & 31\,218 & 31\,178 & 19.1 & 200.6 & 1\,638.0 & 6\,875.5\tabularnewline
 & & 261\,787 &  & $\circ$ & 42\,193 & 31\,092 & 31\,042 & 22.3 & 399.8 & 1\,609.6 & 8\,651.8\tabularnewline
 \midrule
 \multirow{4}{*}{\begin{sideways}\hspace{-0.5em}Generated\end{sideways}} & Gen. & 100 & \multirow{2}{*}{42} & $\bullet$ & 200 & 158 & 158 & 4.6 & 0.2 & 3.5 & 4.1\tabularnewline
 & 160K & 930 &  & $\circ$ & 193 & 158 & 158 & 6.1 & 1.0 & 3.3 & 4.9\tabularnewline[0.3em]
 & Gen. & 1\,000 & \multirow{2}{*}{0.391} & $\bullet$ & 1.161 & 0.395 & 0.395 & 3.0 & 3.3 & 43.8 & 43.8\tabularnewline
 & 0.4K & 10\,649 &  & $\circ$ & 182 & 5.52 & 5.52 & 6.1 & 15.9 & 52.9 & 78.8\tabularnewline
\end{tabular}
\end{table}

For the results in Table~\ref{tab:real-world} we used two Facebook graphs \cite{tmp-sf-12} and five SNAP graphs \cite{lk-snapd-14} as social networks and four web graphs from the 10th DIMACS Implementation Challenge \cite{bmsw-gpgcd-13,bcsv-ucasf-04,brsv-llpam-11,bv-twgfi-04}. We evaluate two variants of QTM. The first is the standard variant which starts with a non-trivial skeleton obtained by the heuristic described in Section~\ref{sec:linear_editing}. The second variant starts with a trivial skeleton where every node is a root. We chose these two variants to determine which part of our algorithm has which influence on the final result. For the standard variant we report the number of edits needed before any node is moved. With a trivial skeleton this number is meaningless and thus we report the number of edits after one round. All other measures are straightforward and are explained in the table's caption.

Even though for some of the graphs the mover needs more than 20 iterations to terminate, the results do not change significantly compared to the results after round 4.
In practice we can thus stop after 4 rounds without incurring a significant quality penalty.
It is interesting to see that for the social networks the initialization algorithm sometimes produces a skeleton that induces more than $m$ edits (e.g.\ in the case of the ``Penn'' graph) but still the results are always slightly better than with a trivial initial skeleton.
This is even true when we do not abort moving after 4 rounds.
For the web graphs, the non-trivial initial skeleton does not seem to be useful for some graphs.
It is not only that the initial number of edits is much higher than the finally needed number of edits, also the number of edits needed in the end is slightly higher than if a trivial initial skeleton was used.
This might be explained by the fact that we designed the initialization algorithm with social networks in mind.
Initial skeleton heuristics built specifically for web graphs could perform better.
While the QTM algorithm needs to edit between approximately 50 and 80\% of the edges of the social networks, the edits of the web graphs are only between 10 and 25\% of the edges.
This suggests that quasi-threshold graphs might be a good model for web graphs while for social networks they represent only a core of the graph that is hidden by a lot of noise.
Concerning the running time one can clearly see that QTM is scalable and suitable for large real-world networks.

As we cannot show for our real-world networks that the edit distance that we get is close to the optimum we generated graphs by generating quasi-threshold graphs and applying random edits to these graphs.
The details of the generation process are described in the appendix.
In Table~\ref{tab:real-world} we report the results of two of these graphs with $400$ and $160\,000$ random edits.
In both cases the number of edits the QTM algorithm finds is below or equal to the generated editing distance.
If we start with a trivial skeleton, the resulting edit distance is sometimes very high, as can be seen for the graph with 400 edits.
This shows that the initialization algorithm from Section~\ref{sec:linear_editing} is necessary to achieve good quality on graphs that need only few edits.
As it seems to be beneficial for most graphs and not very bad for the rest, we suggest to use the initialization algorithm for all graphs.

\begin{wrapfigure}[15]{o}{0.4\textwidth}
\includegraphics[width=\linewidth]{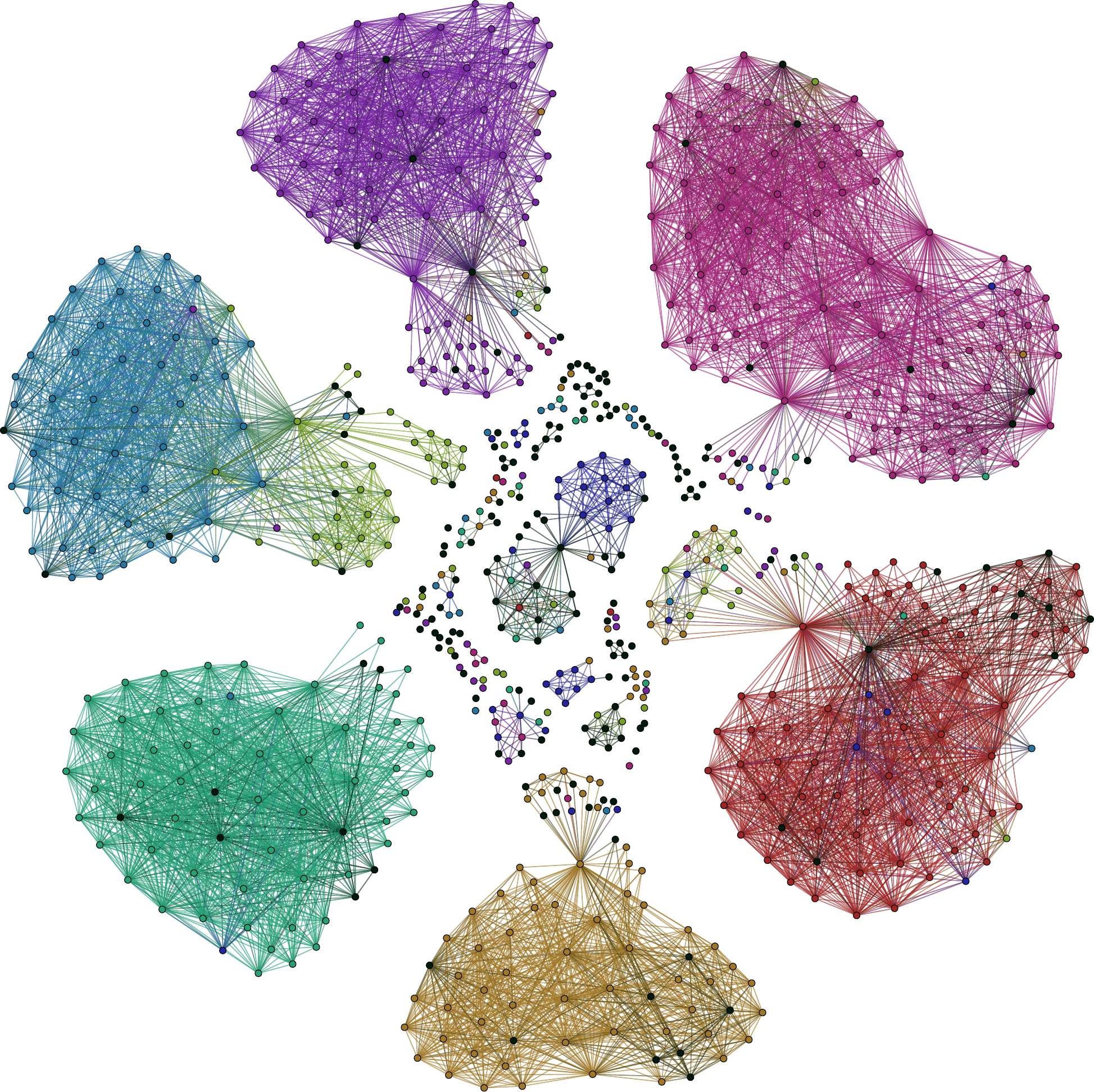}
\caption{Edited Caltech network, edges colored by dormitories of endpoints.}
\label{fig:caltech}
\end{wrapfigure}

\paragraph{Case Study: Caltech.} The main application of our work is community detection.
While a thorough experimental evaluation of its usefulness in this context is future work we want to give a promising outlook.
Figure~\ref{fig:caltech} depicts the edited Caltech university Facebook network from~\cite{tmp-sf-12}.
Nodes are students and edges are Facebook-friendships.
The dormitories of most students are known.
We colored the graph according to this ground-truth.
The picture clearly shows that our algorithm succeeds at identifying most of this structure.

\section{Conclusion}

We have introduced Quasi-Threshold Mover (QTM), the first heuristic algorithm to solve the quasi-threshold editing problem in practice for large graphs.
As a side result we have presented a simple certifying linear-time algorithm for the quasi-threshold recognition problem.
A variant of our recognition algorithm is also used as initialization for the QTM algorithm.
In an extensive experimental study with large real world networks we have shown that it scales very well in practice.
We generated graphs by applying random edits to quasi-threshold graphs.
QTM succeeds on these random graphs and often even finds other quasi-threshold graphs that are closer to the edited graph than the original quasi-threshold graph.
A surprising result is that web graphs are much closer to quasi-threshold graphs than social networks, for which quasi-threshold graphs were introduced as community detection method.
A logical next step is a closer examination of the detected quasi-threshold graphs and the community structure they induce.
Further our QTM algorithm might be adapted for the more restricted problem of threshold editing which is NP-hard as well.\footnote{P{\aa}l Gr{\o}n{\aa}s Drange, personal communication (planned ESA submission)}

\paragraph{\textbf{Acknowledgment:}} We thank James Nastos for helpful discussions.

\printbibliography[segment=1]
\end{refsegment}

\newpage
\begin{refsegment}

\begin{appendix}

\section{Fast Computation of Lower Bounds}

As outlined in the main paper, the idea for computing lower bounds is to find a $C_4$ or $P_4$ and to remove two of the nodes, two neighboring nodes in the case of a $C_4$, the two central nodes in the case of the $P_4$, such that we destroy as few $C_4$ and $P_4$ as possible.
Finding a single $P_4$ or $C_4$ is possible in linear time using the certifying recognition algorithm.
The challenge when designing a fast algorithm for computing lower bounds is that the lower bound can be $n/2$ which results in a quadratic algorithm.

It is enough if we identify all central edges of a $P_4$ as we only want to remove the two nodes that are incident to that edge anyway.
For a $C_4$ it is also enough if we can identify any edge it is part of as we also just want to remove the two incident nodes.
Therefore it is enough if we can quickly find an edge that is part of a $C_4$ or a central edge of a $P_4$.

If we consider the neighbors of the two nodes $u$ and $v$ and want to find a $C_4$ or a $P_4$ where $\{u, v\}$ is the central edge, we only need to find two nodes $x \in N(u) \setminus \{v\}$ and $y \in N(v) \setminus \{u\}$ such that $x \notin N(v)$ and $y \notin N(u)$.
Common neighbors of $u$ and $v$ thus cannot be chosen for $x$ and $y$, however all other neighbors besides $u$ and $v$ can be chosen.
Therefore we know that such two nodes exist whenever $\PseudoCount(\{u, v\}) =   (d(u) - 1 - t(\{u, v\}))\cdot(d(v) - 1 - t(\{u, v\})) > 0$.
The algorithm we choose is based on this observation.
Initially, we count the triangles per edge for all edges.
Then we iterate over all nodes and for each node $u$ we choose a neighbor $v$ such that $\PseudoCount(\{u, v\}) > 0$ and remove $u$ and $v$.
After removing $u$ and $v$ we update the triangle counters accordingly.

In order to destroy not too many $P_4$ and $C_4$, we sort the nodes initially be degree in ascending order.
We also choose the neighbor $v$ of $u$ such that the degree of $v$ is minimal.
Note that the initial iteration order does not necessarily reflect the degree order anymore after removing some of the nodes.

Given a graph structure that allows removing a node $u$ in amortized time $O(d(u))$ the whole algorithm can actually be implemented in time $O(\alpha(G)m)$ with $O(\alpha(G)m)$ memory consumption.
The running time $O(\alpha(G)m)$ comes from triangle listing \cite{cn-asla-85}.
The main idea is that we store for each edge the pairs of edges which form a triangle.
Whenever we delete a node, we check for each edge for all stored pairs if the two other edges still exist, and if yes, decrease their counter.
As we delete each edge only once this gives a total running time of $O(\alpha(G)m)$.

However in practice we found that the required amount of storage was too high for our compute servers.
Even if we have ``just'' 40 triangles per edge on average (for example the web graph ``eu-2005'' from the 10th DIMACS Implementation Challenge \cite{bmsw-gpgcd-13}) storing these triangles means that we need a lot more memory than for just storing $G$.

Therefore we used the trivial update algorithm that, for deleting the edge $\{u, v\}$, enumerates all triangles the edge is part of and updates the counters accordingly.
For deleting all edges this gives a $O(m \cdot \Delta)$ algorithm which only needs $O(m)$ memory.
In practice this was still fast enough for the graphs we considered.
In Table~\ref{tab:bounds} we report the lower bound and the running time of the lower bound calculation for the large real-world graphs we considered (refer to the experimental evaluation, Section~\ref{sec:exp} for details concerning the graphs).
The graphs are sorted by the number of edges $m$.
As in the experimental evaluation, we executed all experiments ten times with different random node id permutations.
Only for the largest graph, uk-2002, we used only one run with the original node ids for the lower bound calculation due to memory constraints.
In this table we report average and maximum bound and running time while in the experimental evaluation section we only reported the maximum.
However, as one can see, the average and the maximum do not differ significantly.
The running times clearly show that the running time does not only depend on $m$ but also on the degrees, i.e.\ graphs with a lower number of nodes but a comparable number of edges have a higher running time.

\begin{table}
\caption{Results for the lower bounds of the large real-world graphs we considered}
\label{tab:bounds}
\centering
\begin{tabular}{l@{~~}r@{~~}r@{~}|@{~}r@{~~}r@{~}|@{~}r@{~~}r}
        Name &          $n$ &           $m$ & \multicolumn{2}{c|@{~}}{Lower Bound} &            \multicolumn{2}{c}{Time [s]}       \\
             &              &        &        mean &      max &                 mean &  max \\
\midrule
   Caltech36 &      769 &     16\,656 &       349.5 &     350 &                  0.0 &   0.0 \\
  com-amazon &   334\,863 &    925\,872 &     99\,305.9 &   99\,413 &                  0.6 &   0.6 \\
    com-dblp &   317\,080 &   1\,049\,866 &     53\,656.9 &   53\,680 &                  0.7 &   0.7 \\
      Penn94 &    41\,554 &   1\,362\,229 &     19\,918.7 &   19\,920 &                  1.8 &   1.8 \\
    cnr-2000 &   325\,557 &   2\,738\,969 &     48\,500.0 &   48\,739 &                 22.6 &  23.8 \\
 com-youtube &  1\,134\,890 &   2\,987\,624 &    139\,006.5 &  139\,077 &                  9.7 &   9.8 \\
     in-2004 &  1\,382\,908 &  13\,591\,473 &    194\,849.9 &  195\,206 &                 70 &  70.4 \\
     eu-2005 &   862\,664 &  16\,138\,468 &    228\,457.1 &  228\,759 &                187.2 & 188.1 \\
      com-lj &  3\,997\,962 &  34\,681\,189 &   1\,334\,663.3 & 1\,334\,770 &                 65.8 &  66.3 \\
   com-orkut &  3\,072\,441 & 117\,185\,083 &   1\,479\,977.2 & 1\,480\,007 &                394.2 & 395.4 \\
     uk-2002 & 18\,520\,486 & 261\,787\,258 &        & 2\,966\,359 &                  & 960.8 \\
\end{tabular}
\end{table}

\section{Details of the Initialization Algorithm}

In Algorithm~\ref{alg:full_init} we provide the full initialization heuristic as pseudo code.
Note that while for the parent calculation we use $\le$ for comparisons we use $<$ for the final selection of the neighbors to keep in order to not to wrongly assign too many neighbors to $u$.

\begin{algorithm}
	\KwIn{$G = (V, E)$}
	\KwOut{Parent assignment $p$ for each node}

	Sort $V$ by degree in descending order using bucket sort\;
	$p : V \rightarrow V \cup \{ \emptyset \}, u \mapsto \emptyset $\;

	Count triangles $t(\{u, v\})$\;

	\ForEach{$u \in V$}{
		\tcp*[h]{Process node $u$}

		$N \gets \{v \in N(u) \,|\, v $ not processed and $ p(u) = p(v)$ or ($\PseudoCount(\{u, v\}) \le  \PseudoCount(\{v, p(v)\})$ and $\mathrm{depth}(v) \le t(\{u, v\}) + 1)\}$\;
		$p_n \gets$ the most frequent value of $p(x)$ for $x \in N$\;

		\If{$p_n \neq p(u)$}{
			$p(u) \gets p_n$\;
			$\mathrm{depth}(u) \gets 0$\;
			$\PseudoCount(\{u, p_n\}) \gets \infty$\;
		}

		\ForEach{$v \in N(u)$ that has not been processed}{
			\If{$p(u) = p(v)$ or $(\PseudoCount(\{u, v\}) <  \PseudoCount(\{v, p(v)\})$ and $\mathrm{depth}(v) < t(\{u, v\}) + 1)$}{
				$p(v) \gets u$\;
				$\mathrm{depth}(v) \gets \mathrm{depth}(v) + 1$\;
			}
		}
	}
	\caption{The Initialization Algorithm}
	\label{alg:full_init}
\end{algorithm}

\section{The Quasi-Threshold Mover in Detail}

Here we want to describe the quasi-threshold mover algorithm in more detail.
Apart from giving more details how we actually implemented the algorithm in order to achieve the claimed running time we will also give proofs for its correctness and running time.

The QTM algorithm iteratively modifies the forest that defines a quasi-threshold graph.
In the following we assume again that our forest has a virtual root $r$ that is connected to all nodes in the original graph, i.e.\ we consider only the case of a tree.

For a single node $v_m$ the algorithm solves the following problem optimally:

Find a parent $u$ in the forest and a set of close children $C$ of that parent $u$ such that inserting $v_m$ as child of $u$ and moving $C$ to be children of $v_m$ minimizes the number of edits among all choices of $u$ and $C$.

One iteration of the algorithm consists of solving this problem for every node of the graph.
We will show later that for a single node $v_m$ this is possible in time $O(d(v_m) \log(d(v_m)))$ time amortized  over an iteration, so the time for a whole iteration is in $O(n + m \log (\Delta))$.
For $t > 0$ iterations the total running time is $O(t \cdot (n + m \log (\Delta)))$.

The main idea why this works is that we do not need to consider all possible parents but only those parents which are adjacent to $v_m$ or which have a close child, i.e.\ a child of which more than half of the descendants are adjacent to $v_m$.
Otherwise the existing edges do not compensate for the missing edges and we could as well add $v_m$ as child of $r$, i.e.\ delete all edges in the original graph that are incident to $v_m$.
We will show how to determine these possible parents and close children by visiting only a constant number of nodes for each neighbor of $v_m$ that are determined by populating counters from the bottom to the top of the tree.

\begin{algorithm}[tb]
	\caption{Core algorithm of QTM: finding a new parent and children to be adopted.}
	\label{alg:local_moving}
	\CommentSty{// Assumption: $v_m$ is not in the tree}\;
	Insert neighbors of $v_m$ in the queue\; \nllabel{lm:pqinit}

	\While{Queue is not empty}{
		$u \gets$ pop node from queue\;
		mark $u$ as touched\; \nllabel{lm:touched}

		\lIf{$\ChildCloseness(u) > \BestSol(u)$}{$\BestSol(u) \gets \ChildCloseness(u)$}\nllabel{lm:gupdate}

		\If{$u$ is marked as neighbor}{$\ChildCloseness(u) \gets \ChildCloseness(u) + 2$, $\BestSol(u) \gets \BestSol(u) + 2$\;} \nllabel{lm:neighborupdate}

		$\ChildCloseness(u) \gets \ChildCloseness(u) - 1$,
		$\BestSol(u) \gets \BestSol(u) - 1$\;\nllabel{lm:tokendec}

		\If(\CommentSty{// Start a DFS from $u$}){$\ChildCloseness(u) \ge 0$
and $u$ has children}{ \nllabel{lm:dfsstart}
			$x \gets $ first child of $u$\;
			\While{$x \neq u$}{\nllabel{lm:dfswhile}
				\eIf{$x$ not touched or $\ChildCloseness(x) < 0$}{
					$\ChildCloseness(u) \gets \ChildCloseness(u) - 1$\;
					$x \gets \DFS(x)$\;
					\If{$\ChildCloseness(u) <
0$}{\nllabel{lm:dfsend}
						$\DFS(u) \gets x$\;
						break\;
					}
					$x \gets$ next node in DFS
order after $x$ below $u$\;
				}{
					$x \gets$ next node in DFS order
after the subtree of $x$ below $u$\;\nllabel{lm:dfsskip}
				}
			}
		}

		\If(\CommentSty{// Propagate information to parent}){$u
\neq r$}{
			\If{$\ChildCloseness(u) > 0$}{
				$\ChildCloseness(p(u)) \gets \ChildCloseness(p(u)) + \ChildCloseness(u)$\;
				Insert $p(u)$ in queue\;\nllabel{lm:einsertpq}
			}

			\If{$\BestSol(u) > \BestSol(p(u))$}{
				$\BestSol(p(u)) \gets \BestSol(u)$\;
				Insert $p(u)$ in queue\;
			}
		}
	}
\end{algorithm}

In one iteration, the QTM algorithm simply iterates over all nodes in a random order.
For each node $v_m$ we search the optimal parent and children.
Algorithm~\ref{alg:local_moving} contains the pseudo code for the main part of this search.

In order to avoid complicated special cases we first remove $v_m$ from the tree\footnote{This is equivalent to isolating $v_m$ by inserting $v_m$ below the virtual root $r$.}
In the end we want to move $v_m$ back to its initial position if no better position was found.
If the initial position of $v_m$ was the best position, then the algorithm will find it again.
However, if there are multiple positions in the skeleton that induce the same number of edits, the algorithm will find any of these positions.
In order to make sure that the algorithm terminates even if we do not limit the number of iterations we store the initial position of $v_m$, i.e.\ its children and its parent.
We also count the number of edits that were necessary among its neighbors.
If no improvement was possible, we move the node back to this initial position in the end.

For a single node $v_m$ that shall possibly be moved we will process its neighbors and possibly $O(d(v_m))$ other nodes ordered by decreasing depth.
We maintain the list of these nodes in a priority queue that is initialized with the neighbors of $v_m$ and sorted by depth.
As we do not want to dynamically determine the depth of a node we calculate the depth initially and update it whenever we remove or insert a node in the forest.
A marker is set for all neighbors of $v_m$ in order to make it possible to determine in constant time if a node is adjacent to $v_m$.

When we process a node $u$ of the queue, we first determine if $u$ is the best parent in the subtree of $u$, then we possibly visit some nodes below $u$ using a special DFS in order to determine the child closeness of $u$ and possibly insert its parent into the queue.
We will later explain the details of the DFS.

We store the score of the best solution in the subtree of $u$ in $\BestSol(u)$ and the child closeness of $u$ in $\ChildCloseness(u)$. In order to avoid special cases we initialize $\BestSol(u)$ with $-1$ and $\ChildCloseness(u)$ with $0$.
Furthermore we store at each node $u$ the state of the DFS that has possibly been started at $u$.
In order to store the state we only store the last visited node.
We store this node in $\DFS(u)$.
We initialize $\DFS(u)$ with $u$.

At the end, we can find the number of edits that can be saved over isolating $v_m$ in $\BestSol(r)$ and we can also additionally track which parent lead to that score.
As already mentioned, we compare this to the number of edits at the old position of $v_m$ and move $v_m$ back to the old position if no improvement was possible.
If an improvement is possible, we insert $v_m$ below the parent that we identified as best parent $u$.
The missing part are the children that shall be moved from $u$ to $v_m$.
We can determine them by visiting all previously visited nodes (we can store them) and check for each visited node $c$ if it is a close child of $u$, i.e.\ if attaching it to $v_m$ would save edits which we have stored in $\ChildCloseness(c)$.

\paragraph{Proof of Correctness}

In this section we want to give a formal proof why the local moving algorithm is correct, i.e.\ always selects the best parent and the best selection of children.
We do this by giving exact definitions of all used variables and proofing their correctness.

We begin with the child closeness $\ChildCloseness(u)$ which is the number of edits that we can save if $v_m$ is attached below $u$:

\begin{proposition}
\label{prop:e_correct}
Either $\ChildCloseness(u)$ is the number of neighbors of $v_m$ in the subtree of $u$ minus the
number of non-neighbors, or there are more non-neighbors than neighbors in the
subtree of $u$. In the latter case, if $u$ has been processed, then $\ChildCloseness(u) = -1$.
More precisely if $u$ has been processed, $\ChildCloseness(u)$ is the number of existing neighbors minus the number of missing neighbors of all nodes in DFS order between $u$ and $D(u)$ and additionally all subtrees of children $c$ with $\ChildCloseness(c) > 0$ that are not in the DFS order between $u$ and $\DFS(u)$.
\end{proposition}

\begin{proof}
We will give the proof by structural induction.

As first step we want to establish that all nodes where $\ChildCloseness(u) \ge 0$ are processed.
As $\ChildCloseness(u) < 0$ if there are no neighbors of $v_m$ in the subtree of $u$ only neighbors of $v_m$ and their ancestors can have $\ChildCloseness(u) \ge 0$.
All neighbors of $v_m$ are processed (line~\ref{lm:pqinit}).
For non-neighbors $u$, $\ChildCloseness(u) \ge 0$ means that one of their children $c$ is close, i.e.\ has $\ChildCloseness(c) > 0$.
As in this case $c$ inserts $v_m$ into the queue (line~\ref{lm:einsertpq}) also in this case $u$ will be processed.
As we process all nodes by descending depth (only parents, i.e.\ nodes of smaller depth, are inserted in the queue) we can assume that if we are at a node $u$, all descendants of $u$ that need to be processed have been processed and that when the algorithm terminates all nodes $u$ with $\ChildCloseness(u) \ge 0$ have been processed.

In line~\ref{lm:neighborupdate} and \ref{lm:tokendec} $\ChildCloseness(u)$ is updated such that the proposition is true if we consider only $u$ itself.
This means that the proposition is true for leafs which is also the initial step of our induction.

As the claim is true for all children of $u$ we can also assume that all children $c$ with $\ChildCloseness(c) > 0$ already updated $\ChildCloseness(u)$ accordingly, i.e.\ $\ChildCloseness(u)$ already correctly considers $u$ and the values of all children with $\ChildCloseness(c) > -1$.

If we have $\ChildCloseness(u) = -1$ in line~\ref{lm:dfsstart} $\ChildCloseness(u)$ must have been 0 initially as in the following it can only be decreased by 1 at maximum.
Therefore we are in the situation that $u$ is no neighbor of $v_m$ and $u$ has no close children, i.e.\ children with $\ChildCloseness(c) > 0$.
In this situation this is already the final result as if this result was incorrect, i.e.\ $\ChildCloseness(u) > -1$, then there must be at least as many neighbors as non-neighbors of $v_m$ among the nodes in the subtree of $u$.
As $u$ is no neighbor of $v_m$ the descendants must contain at least one more neighbor than there are non-neighbors among them and this must also be true for at least one of the children of $u$.
Therefore $\ChildCloseness(c) > 0$ for this child which is contradiction to the situation that there are no children $c$ with $\ChildCloseness(c) > 0$.

So now we only need to consider the case that $\ChildCloseness(u) > -1$ in line~\ref{lm:dfsstart} which means that the algorithm will start a DFS.

As first step we want to have a closer look at the DFS that is executed in the algorithm.
If we say in the following that the DFS ``visits'' a node we mean that it is the value of $x$ in line~\ref{lm:dfswhile}.

On visiting certain nodes, we decrease $\ChildCloseness(u)$.
If $\ChildCloseness(u) < 0$, we stop the DFS and store the last visited node in $\DFS(u)$.
This means that at the end of a DFS either $\ChildCloseness(u) \ge 0$ or $\DFS(u)$ points to the last visited node.
Whenever we visit a node, there are three possible cases:

\begin{enumerate}
\item The easiest case is that $c$ has been processed and $\ChildCloseness(c) > -1$. In this
case the edits of $c$ are already considered by the parent of $c$ and we do
not need to deal with it. Furthermore, we know that $\ChildCloseness(c)$ is the correct number
of neighbors minus non-neighbors of the whole subtree of $c$ so it is correct
that the algorithm skips these nodes in line~\ref{lm:dfsskip}.

\item If $c$ has not been processed yet, it is a not a neighbor of $v$
(otherwise it would have been processed). We decrease $\ChildCloseness(u)$ which is correct as
this is a missing neighbor. Then we can continue the DFS. The same is true if
$c$ has been processed, $\ChildCloseness(c) < 0$ but $\DFS(c) = c$, i.e.\ no DFS has been
executed.

\item If $\ChildCloseness(c) < 0$ and $\DFS(c) \neq c$ we know from the induction
hypothesis that $\ChildCloseness(c) = -1$ and that this is exactly the number of existing
minus the number of missing neighbors from $c$ up to $\DFS(c)$ in DFS order (including $\DFS(c)$)
plus the number of children $c'$ of $c$ with $\ChildCloseness(c') > -1$ which we ignore anyway
in the DFS.
We decrease $\ChildCloseness(u)$ which correctly considers the nodes between $c$ and $\DFS(c)$ in DFS order (both included).
Then we jump to $\DFS(c)$ and do not visit $\DFS(c)$ but the next node in DFS order which is obviously correct as $\DFS(c)$ has already been considered by decreasing $\ChildCloseness(u)$.
\end{enumerate}

When the DFS ends, either we have now considered all edits of the descendants
of $u$ or the DFS ended with $\ChildCloseness(u) < 0$ and we have stored the location of the
last visited node in $\DFS(u)$. In the latter case, all nodes up to this point
have been considered as we have outlined before. Therefore the claim is now
also true for $u$. \qed

\end{proof}

If we want to know for a potential parent $u$ how many edits we can save by moving some of its children to $v_m$ this is the sum of $\ChildCloseness(c)$ for all close children of $u$, i.e.\ children with $\ChildCloseness(c) > 0$.
This is the value that we store in $\ChildCloseness(u)$ before $u$ is processed by setting $\ChildCloseness(p(c))$ for all close children $c$ of $u$, i.e.\ children $c$ with $\ChildCloseness(c) > 0$.
Obviously, this is positive if a node $u$ has at least one close child.

In order to not to need to evaluate all nodes as potential parents we make use of the following observation:

\begin{proposition}
\label{prop:possibleparents}
Only nodes with close children and neighbors of $v_m$ need to be considered as parents of $v_m$.
\end{proposition}

\begin{proof}
Assume otherwise: The best parent $u$ has no close children and is not a neighbor of $v_m$.
Then attaching children of $u$ to $v_m$ makes no sense as this would only increase the number of needed edits so we can assume that no children will be attached.
However then choosing $p(u)$ as parent of $v_m$ will save one edit as $u$ is no neighbor of $v_m$.
This is a contradiction to the assumption that $u$ is the best parent. \qed
\end{proof}

So far we have only evaluated edits below nodes and identified all possible parents which are also processed as we have established before.
The part that is still missing is the evaluation of the edits above a potential parent $u$.

\begin{theorem}
Consider the subtree $T_u$ of $u$.
Then for the subgraph of $T_u$, $\BestSol(u)$ stores the maximum number of edits from $v_m$ to nodes inside the subgraph of $T_u$ that can be saved by choosing the parent of $v_m$ in $T_u$ instead of isolating $v_m$ or $-1$ if no edits can be saved.
\end{theorem}

\begin{proof}
The proof is given by structural induction on the tree skeleton.
We start with the initial step which is a node $u$ that is a leaf of the tree.

As $T_u$ only consists of $u$, we have only one edge from $v_m$ to $u$ and therefore only two cases:
$u$ is a neighbor of $v_m$ or not.
In both cases, $\BestSol(u)$ is initialized with $-1$ as there are no children that could propagate any values.
In the second case, $u$ will not be processed but the result is already correct anyway: no edits can be saved by choosing $u$ as parent of $v_m$.
In the first case, as $\ChildCloseness(u)$ is initialized to $0$, $\ChildCloseness(u) > \BestSol(u)$ and therefore $\BestSol(u) \gets 0$ (line~\ref{lm:gupdate}).
We end with $\BestSol(u) = 1$ which is correct, we can save an edit over isolating $u$ as the edge $(u, v_m)$ does not need to be deleted when we chose $u$ as parent of $v_m$.
When we set $\BestSol(u)$ we can also store $u$ as best parent together with $\BestSol(u)$.

Now we can assume that the theorem holds for all children of $u$.

As not all nodes are processed, we need to explain why $u$ is processed at all if $\BestSol(u) > -1$ should hold.
There are two possibilities: Either it could make sense to use $u$ as parent or we could use a node below $u$ as parent.
In the first case by Proposition~\ref{prop:possibleparents} either $u$ is a neighbor of $v_m$ or $u$ has a close child which means that $u$ is processed.
Assume that in the second case $u$ was not processed but it should be $\BestSol(u) > -1$.
Further we can assume that $u \notin N(v_m)$ as otherwise $u$ was processed.
Let $x$ be the best parent in $T_u$ and let $c$ be the direct child of $u$ such that $x$ is in the subtree of $c$ $T_c$ (it is possible that $x = c$).
If it makes sense to use $x$ as parent of $v_m$ then by inserting $v_m$ below $x$ also the edge $\{u, v_m\}$ must be inserted.
This means that in the subtree of $c$ we can save one more edit as the edit $\{u, v_m\}$ is not necessary which means that $\BestSol(c) = \BestSol(u) + 1 > 0$.
This means that by induction $c$ must have been processed and $c$ must have propagated $\BestSol(c)$ to $u$ and also inserted $u$ in the queue which is a contradiction to the assumption that $u$ is not processed.

When $u$ is processed, we need to make the decision if $u$ is the best parent in $T_u$ or if we should choose the parent below $u$.
The edit of the edge $\{v_m, u\}$ is needed or not independent of the choice of the parent in $T_u$.
Therefore we do not need to reconsider any decisions that were made below $u$.
If we do not choose $u$, then we need to choose the best parent below $u$, i.e.\ the one of the subtree of the child $c$ with the highest value of $\BestSol(c)$.
This is also what the algorithm does by propagating $\BestSol(c)$ to the parent as maximum of $\BestSol(c)$ and $\BestSol(p(c))$.
Therefore $\BestSol(u)$ is initialized to the best solution below $u$.

If we want to determine how good $u$ is as parent, we need to look at the closeness of its children.
More specifically, we can save as many edits as the sum of $\ChildCloseness(c)$ for all close children $c$ of $u$.
This is the value to which $\ChildCloseness(u)$ is initialized by its close children, therefore we only need to compare $\BestSol(u)$ to $\ChildCloseness(u)$.
Therefore it is correct to set $\BestSol(u)$ to $\ChildCloseness(u)$ if $\ChildCloseness(u)$ is larger than $\BestSol(u)$.
After this initial decision, we increase or decrease $\BestSol(u)$ by one depending on whether the edge $\{u, v_m\}$ exists or not, this is obviously correct.
\qed
\end{proof}

As $T_r$ is the whole graph, $\BestSol(r)$ determines the best solution of the whole graph.
Therefore the QTM algorithm optimally solves the problem of finding a new parent and a set of its children that shall be adopted.

\paragraph{Proof of the Running Time}

After showing the correctness of the algorithm, we will now show that the running time is indeed $O(m \log(\Delta))$ per iteration and amortized $O(d \log(d))$ per node.

During the whole algorithm we maintain a depth value for each node that specifies the depth in the forest at which the node is located.
Whenever we move a node, we update these depth values.
This involves decreasing the depth values of all descendants of the node in its original position and increasing the depth values of all descendants of the node at the new position.
Unfortunately it is not obvious that this is possible in the claimed running time as a node $v_m$ might have more than $O(d(v_m))$ descendants.

Note that a node is adjacent to all its descendants and ancestors in the edited graph.
This means that every ancestor or descendant that is not adjacent to the node causes an insert.
Therefore the node must be neighbor of at least half of the ancestors and children after a move operation as otherwise the less than half of the degree deletes are cheaper than the more than half of the degree inserts.
This means that updating the depth values at the destination is possible in $O(d)$ time.

For the update of the values in the original position we need a different, more complicated argument.
First of all we assume that initially the total number of edits never exceeds the number of edges as otherwise we could simply delete all edges and get less edits.
For amortizing the number of needed edits of nodes that have more descendants and ancestors than their degree we give each node tokens for all their neighbors in the edited graph.
As the number of edits is at most $m$ the number of initially distributed tokens is in $O(m)$.
Whenever we move a node $v_m$, it generates tokens for all its new neighbors and itself, i.e. in total at most $2 \cdot d(v_m)$ tokens.
Therefore a node has always a token for each of its ancestors and descendants and can use that token to account for updating the depth of its previous descendants.
In each round only $O(m)$ tokens are generated, therefore updating the depth values of a node is in amortized time $O(d)$ per node and $O(m)$ per iteration.

Using the same argument we can also account for the time that is needed for updating the pointers of each node to its parents and children and for counting the number of initially needed or saved edits.

What we have shown so far means that once we know the best destination we can move a node and update all depth values in time $O(d)$ amortized over an iteration where all nodes are moved.

The remaining claim is that we can determine the new parent and the new children in time $O(d \log(d))$ per node.
More precisely we will show that only $O(d)$ nodes are inserted in the queue and we need amortized constant time for processing a node.
A standard max-heap that needs $O(\log(n))$ time per operation can be used for the implementation of the queue.

All values that are stored per node need to be initialized for the whole iteration.
All nodes whose values are changed, which are exactly the nodes that have been in the queue at some moment, need to be stored so their values can be reset at the end of the processing of a node.

The basic idea of the main proof is that each neighbor of $v_m$ gets four tokens.
This is represented by the fact that we increase $\ChildCloseness(u)$ and $\BestSol(u)$ by 2 for all neighbors $u$ of $v_m$.
When we process a node $u$, one token is consumed if this node is no neighbor of $v_m$, then the DFS consumes tokens of $\ChildCloseness(u)$ and at the end the rest of the tokens are passed to the parent.

Note that all nodes that are processed have $\ChildCloseness(u) > 0$ or $\BestSol(u) > 0$, either initially or after accounting for the fact that they are neighbors of $v_m$.

First of all let us only consider processed nodes $u$ with $\ChildCloseness(u) > 0$ initially or after accounting for the fact that $u$ is a neighbor of $v_m$.
We consume one token for processing this node in line~\ref{lm:tokendec}.
This is for the whole processing of the node apart from the DFS where the accounting is more complicated.
Apart from the DFS only constant work is done per node, so consuming one taken is enough for that.

First of all note that for each visited node in the DFS only a constant amount of work is needed as traversing the tree, i.e.\ possibly traversing a node multiple times can be accounted to the first visit.
Obviously without keeping a stack this needs a tree structure where we can determine the next child $c'$ after a child $c$ of a node $u$ can be determined in constant time.
This can be implemented by storing in node $c$ the position of $c$ in the array (or list) of children in $p(c)$.
This also allows deleting entries in the children list in constant time (in an array deletion can be implemented as swap with the last child).

Whenever we visit a node that has not been touched yet or that has $\ChildCloseness(x) < 0$, we consume one token of $\ChildCloseness(u)$.
When this is not the case, i.e.\ $\ChildCloseness(x) > -1$, the node has been processed already and we account our visiting of $x$ to the processing of $x$.
This is okay as we visit each node only once during a DFS:
After the DFS starting at $u$ has finished, either $\ChildCloseness(u) > -1$ and an upcoming DFS will not descend into the subtree of $u$ anymore or we ended the DFS in line~\ref{lm:dfsend} and thus have set $\DFS(u)$ to the last visited node which means that when we visit $u$ in an upcoming DFS, this DFS will directly jump to $\DFS(u)$ after visiting $u$.

Note that by decreasing $\ChildCloseness(u)$ to $-1$ we actually consume one more token than we had.
However for this we only need a constant amount of work which can be accounted for by the processing time of $u$.

Now we consider nodes $u$ that are processed with $\BestSol(u) > 0$ initially.
If we ignore line~\ref{lm:gupdate} everything seems to be simple: we consume one token and pass the rest to the parent (using the maximum instead of the sum) if a token is left.
However if we set $\BestSol(u)$ to $\ChildCloseness(u)$ we are getting new tokens out of nowhere.
Fortunately it turns out we can explain that these tokens are also from the $\BestSol(c)$ of all children $c$ of $u$ but the sum instead of the maximum:
Note that $\ChildCloseness(c) \le \BestSol(c)$ for any node $c$ that has been processed, i.e.\ after line~\ref{lm:gupdate} $\ChildCloseness(c) \le \BestSol(c)$ holds, then in line~\ref{lm:neighborupdate} both are increased by 2 and after that only $\ChildCloseness(c)$ is decreased.
As $\ChildCloseness(u)$ is initially the sum of all positive $\ChildCloseness(c)$ of the children $c$ of $u$, it follows that initially $\ChildCloseness(u)$ is smaller or equal to the sum of all $\BestSol(c)$ of the children $c$ of $u$.
Therefore actually each child $c$ has passed a part of the tokens of $\BestSol(c)$ to $u$ in form of $\ChildCloseness(u)$.
Therefore also line~\ref{lm:gupdate} does not create new tokens.

This means that in total we only process $O(d)$ nodes and do amortized constant work per node as we have claimed.

\section{Details of the Algorithm proposed by Nastos and Gao}

Nastos and Gao \cite{ng-f-13} describe that in their greedy algorithm they test each possible edge addition and deletion (i.e.\ all $O(n^2)$ possibilities) in order to choose the edit that results in the largest improvement, i.e.\ the highest decrease of the number of induced $P_4$ and $C_4$.
After executing this greedy heuristic they revert the last few edits and execute the bounded search tree algorithm.
If this results in a solution with fewer edits, they repeat this last step until no improvement is possible anymore.
We chose 10 for the number of edits that are reverted.

The main question for the implementation is thus how to select the next edit.
As far as we know it is an open problem if it is possible to determine the edit that destroys most $P_4$ and $C_4$ in time $o(n^2)$.
Therefore we concentrate on the obvious approach that was also implied by Nastos and Gao: execute each possibility and see how the number of $P_4$ and $C_4$ changes.
The main ingredient is thus a fast update algorithm for this counter.

As far as we are aware the fastest update algorithm for counting node-induced $P_4$ and $C_4$ subgraphs needs amortized time $O(h^2)$ for each update where $h$ is the $h$-index of the graph \cite{egst-e-12}.
While the worst-case bound of the $h$-index is $\sqrt{m}$ it has been shown that many real-world social networks have a much lower $h$-index \cite{es-tigia-09}.
However this algorithm requires constant-time edge existence checks and stores many counts for pairs and triples of edges (though only if they are non-zero).

We implemented a different algorithm which has the same worst-case complexity if we ignore the actual value of $h$: $O(m)$.
Furthermore this algorithm is much simpler to implement and while it needs $O(n)$ additional memory during updates only the counter itself needs to be stored between updates.
The initial counting is thus possible in time $O(m^2)$, therefore this results in an $O(m^2 + k \cdot n^2 \cdot m)$ algorithm.
Note that the time needed for the initial counting is dominated by the time needed for each edit.

The main idea of the algorithm is that we examine the neighborhood structure of the edge that shall be deleted or inserted.
Using markers we note which neighbors are common or exclusive to the two incident nodes.
We iterate once over each of these three groups of neighbors and over their neighbors which needs at most $O(m)$ time.
Based on the status of the markers of these neighbors of the neighbors we can count how many times certain structures occur on the neighborhood of the edge.
Using these counts we can determine how many $P_4$ and $C_4$ were destroyed and created by editing the edge.

Apart from applying the update algorithm $m$ times there is also a simpler $O(m^2)$ counting algorithm which we used for the initialization.
Here the idea is again that we determine the common and exclusive neighborhoods for each edge.
Then we only need to iterate over the exclusive neighbors of one of the two nodes and check for each of them how many of its neighbors are exclusive to the other node.
This gives us the number of $C_4$ that edge is part of.
The product of the sizes of the exclusive neighborhoods gives us the number of $P_4$ where the edge is the central edge plus the number of $C_4$ the edge is part of.
Combining both we can get the number of $P_4$ where the edge is the central edge.
While the sum of these values already gives the number of $P_4$, the sum of the $C_4$-counts still needs to be divided by 4.
Note that when the graph is a quasi-threshold graph, i.e.\ there are not $P_4$ and $C_4$, this needs only $O(m\cdot\Delta)$ time.

\section{Generated Graphs}

Each connected component of the quasi-threshold graph was generated as reachability graph of a rooted tree.
For generating a tree, $0$ is the root and each node $v \in \{1, \dots, n-1\}$ chooses a parent in $\{0, \dots, v-1\}$.

As shown by \cite{lksf-ccscn-10} many real-world networks including social networks exhibit a community size distribution that is similar to a power law distribution.
Therefore we chose a power law sequence with 10 as minimum, $0.2 \cdot n$ as maximum and $-1$ as exponent for the component sizes and generated trees of the respective sizes.

For $k$ edits we inserted $0.8 \cdot k$ new edges and deleted $0.2 \cdot k$ old edges of the quasi-threshold graph chosen uniformly at random.
Therefore after these modifications the maximum editing distance to the original graph is $k$.
We used a more insertions than deletions as preliminary experiments on real-world networks showed that during editing much more edges are deleted than inserted.

\begin{table}
\caption{Results for the generated graphs}
\label{tab:genresults}
\centering
\begin{tabular}{l@{~}r@{~~}r@{~~}r@{~}|@{~}r@{~~}r@{~~}r@{~}|@{~}r@{~}|@{~}r@{~~}r@{~~}r}
Rand & $n$ & $b$ & I & \multicolumn{3}{c|@{~}}{Edits} & It & \multicolumn{3}{c}{Time {[}s{]}}\tabularnewline
Ed. & $m$ &  &  & 0 & 4 & $\infty$ & $\infty$ & 0 & 4 & $\infty$\tabularnewline
\midrule
\multirow{2}{*}{20} & 100 & \multirow{2}{*}{16} & $\bullet$ & 34.4 & 20.0 & 20.0 & 2.4 & 0.0 & 0.0 & 0.0\tabularnewline
 & 269 &  & $\circ$ & 36.7 & 21.4 & 21.4 & 3.9 & 0.0 & 0.0 & 0.0\tabularnewline[0.3em]
\multirow{2}{*}{400} & 100 & \multirow{2}{*}{49} & $\bullet$ & 420.5 & 352.0 & 351.9 & 3.9 & 0.0 & 0.0 & 0.0\tabularnewline
 & 497 &  & $\circ$ & 372.0 & 363.5 & 363.5 & 3.9 & 0.0 & 0.0 & 0.0\tabularnewline[0.3em]
\multirow{2}{*}{20} & 1\,000 & \multirow{2}{*}{19} & $\bullet$ & 38.0 & 19.0 & 19.0 & 2.0 & 0.0 & 0.0 & 0.0\tabularnewline
 & 4\,030 &  & $\circ$ & 166.4 & 21.7 & 21.7 & 4.1 & 0.0 & 0.0 & 0.0\tabularnewline[0.3em]
\multirow{2}{*}{400} & 1\,000 & \multirow{2}{*}{225} & $\bullet$ & 585.3 & 391.2 & 391.2 & 3.4 & 0.0 & 0.0 & 0.0\tabularnewline
 & 4\,258 &  & $\circ$ & 594.4 & 393.6 & 393.6 & 4.4 & 0.0 & 0.0 & 0.0\tabularnewline[0.3em]
\multirow{2}{*}{8K} & 1\,000 & \multirow{2}{*}{494} & $\bullet$ & 8\,268 & 7\,219 & 7\,218.5 & 5.2 & 0.0 & 0.0 & 0.0\tabularnewline
 & 8\,818 &  & $\circ$ & 7\,647 & 7\,511 & 7\,490.6 & 8.3 & 0.0 & 0.0 & 0.0\tabularnewline[0.3em]
\multirow{2}{*}{20} & 10\,000 & \multirow{2}{*}{20} & $\bullet$ & 47.8 & 20.0 & 20.0 & 2.0 & 0.0 & 0.1 & 0.1\tabularnewline
 & 66\,081 &  & $\circ$ & 1\,669 & 69.8 & 69.8 & 4.6 & 0.1 & 0.2 & 0.2\tabularnewline[0.3em]
\multirow{2}{*}{400} & 10\,000 & \multirow{2}{*}{366} & $\bullet$ & 849.3 & 390.6 & 390.6 & 3.2 & 0.0 & 0.2 & 0.2\tabularnewline
 & 66\,309 &  & $\circ$ & 2\,143 & 440.7 & 440.7 & 4.8 & 0.1 & 0.2 & 0.2\tabularnewline[0.3em]
\multirow{2}{*}{8K} & 10\,000 & \multirow{2}{*}{3\,256} & $\bullet$ & 11\,626 & 7\,902 & 7\,902 & 3.9 & 0.0 & 0.2 & 0.2\tabularnewline
 & 70\,869 &  & $\circ$ & 10\,184 & 7\,912 & 7\,911 & 5.1 & 0.1 & 0.2 & 0.3\tabularnewline[0.3em]
\multirow{2}{*}{160K} & 10\,000 & \multirow{2}{*}{4\,985} & $\bullet$ & 157\,114 & 144\,885 & 144\,880 & 5.8 & 0.0 & 0.6 & 0.8\tabularnewline
 & 162\,069 &  & $\circ$ & 150\,227 & 148\,206 & 147\,892 & 9.7 & 0.1 & 0.5 & 1.2\tabularnewline[0.3em]
\multirow{2}{*}{20} & 100\,000 & \multirow{2}{*}{20} & $\bullet$ & 64.3 & 20.0 & 20.0 & 2.0 & 0.2 & 1.7 & 1.7\tabularnewline
 & 833\,565 &  & $\circ$ & 17\,785 & 529.9 & 529.6 & 5.3 & 0.8 & 2.8 & 3.7\tabularnewline[0.3em]
\multirow{2}{*}{400} & 100\,000 & \multirow{2}{*}{384} & $\bullet$ & 1\,047 & 391.3 & 391.3 & 3.2 & 0.2 & 2.6 & 2.6\tabularnewline
 & 833\,793 &  & $\circ$ & 18\,319 & 900.0 & 899.4 & 5.5 & 0.8 & 2.9 & 3.9\tabularnewline[0.3em]
\multirow{2}{*}{8K} & 100\,000 & \multirow{2}{*}{6\,519} & $\bullet$ & 18\,550 & 7\,889 & 7\,889 & 3.4 & 0.2 & 2.8 & 2.8\tabularnewline
 & 838\,353 &  & $\circ$ & 26\,144 & 8\,381 & 8\,381 & 5.4 & 0.8 & 2.8 & 3.8\tabularnewline[0.3em]
\multirow{2}{*}{160K} & 100\,000 & \multirow{2}{*}{42\,021} & $\bullet$ & 199\,558 & 158\,021 & 158\,021 & 4.6 & 0.2 & 3.5 & 4.1\tabularnewline
 & 929\,553 &  & $\circ$ & 193\,071 & 158\,031 & 158\,025 & 6.1 & 1.0 & 3.3 & 4.9\tabularnewline[0.3em]
\multirow{2}{*}{3.2M} & 100\,000 & \multirow{2}{*}{49\,913} & $\bullet$ & 2\,728\,804 & 2\,647\,566 & 2\,647\,564 & 5.8 & 1.1 & 12.5 & 16.8\tabularnewline
 & 2\,753\,553 &  & $\circ$ & 2\,655\,538 & 2\,654\,738 & 2\,654\,736 & 5.7 & 3.0 & 11.9 & 16.9\tabularnewline[0.3em]
\multirow{2}{*}{20} & 1\,000\,000 & \multirow{2}{*}{20} & $\bullet$ & 68.9 & 20.0 & 20.0 & 2.2 & 3.6 & 32.5 & 32.3\tabularnewline
 & 10\,648\,647 &  & $\circ$ & 181\,540 & 5\,116 & 5\,111 & 6.0 & 16.4 & 54.3 & 79.7\tabularnewline[0.3em]
\multirow{2}{*}{400} & 1\,000\,000 & \multirow{2}{*}{391} & $\bullet$ & 1\,161 & 395.1 & 395.1 & 3.0 & 3.3 & 43.8 & 43.8\tabularnewline
 & 10\,648\,875 &  & $\circ$ & 182\,248 & 5\,523 & 5\,518 & 6.1 & 15.9 & 52.9 & 78.8\tabularnewline[0.3em]
\multirow{2}{*}{8K} & 1\,000\,000 & \multirow{2}{*}{7\,447} & $\bullet$ & 25\,085 & 7\,912 & 7\,912 & 3.5 & 3.4 & 50.1 & 50.0\tabularnewline
 & 10\,653\,435 &  & $\circ$ & 189\,504 & 13\,006 & 13\,001 & 6.0 & 16.6 & 53.4 & 78.0\tabularnewline[0.3em]
\multirow{2}{*}{160K} & 1\,000\,000 & \multirow{2}{*}{112\,814} & $\bullet$ & 369\,501 & 158\,808 & 158\,808 & 4.1 & 3.5 & 57.5 & 58.8\tabularnewline
 & 10\,744\,635 &  & $\circ$ & 346\,462 & 163\,337 & 163\,330 & 6.4 & 17.4 & 54.8 & 84.5\tabularnewline[0.3em]
\multirow{2}{*}{3.2M} & 1\,000\,000 & \multirow{2}{*}{476\,562} & $\bullet$ & 3\,747\,793 & 3\,163\,277 & 3\,163\,273 & 5.8 & 4.7 & 71.8 & 101.2\tabularnewline
 & 12\,568\,635 &  & $\circ$ & 3\,820\,935 & 3\,164\,175 & 3\,163\,848 & 7.5 & 26.2 & 78.2 & 134.7\tabularnewline
\end{tabular}
\end{table}

In Table~\ref{tab:genresults} we show the results for all graphs that we have generated.
The first column shows the number of random edits we performed.
As already mentioned in the experimental evaluation for all generated graphs the QTM algorithm finds a quasi-threshold graph that is at least as close as the original one.
Omitting the initialization gives much worse results for low numbers of edits and slightly worse results for higher numbers of edits.
The lower bound is relatively close to the generated and found number of edits for low numbers of edits, for very high numbers of edits it is close to its theoretical maximum, $n/2$.

All in all this shows that the QTM algorithm finds edits that are reasonable but it depends on a good initial heuristic.

\printbibliography[segment=2]

\end{appendix}
\end{refsegment}

\end{document}